\documentclass{article}
\usepackage[utf8]{inputenc}
\usepackage[margin=1in]{geometry}

\usepackage{microtype}
\usepackage{graphicx}
\usepackage{subfigure}
\usepackage{booktabs}
\usepackage{natbib}
\setcitestyle{open={(},close={)}}
\usepackage{amsmath}%
\usepackage{color}%
\usepackage{xspace}%
\usepackage{hyperref}%
\hypersetup{%
   breaklinks,%
   colorlinks=true,%
   linkcolor=[rgb]{0.45,0.0,0.0},%
   citecolor=[rgb]{0,0,0.7}
}

\usepackage{amsmath,amsthm,amsfonts,amssymb,mathdots,array,mathrsfs,bm,bbm,stmaryrd,graphicx,subfigure,xcolor}
\usepackage{algpseudocode,algorithm,algorithmicx}
\usepackage[algo2e,boxed,noend]{algorithm2e}
\usepackage{breakcites}

\usepackage[T1]{fontenc}
\usepackage{enumerate}
\usepackage{inputenc}

\usepackage{graphicx} 
\usepackage{subfigure}

\usepackage{booktabs,balance}
\usepackage{rotating}
\usepackage{boldline}
\usepackage{makecell}
\usepackage{multirow}
\usepackage{balance}

\usepackage{tikz}

\newtheorem{theorem}{Theorem}[section]

\newtheorem{lemma}[theorem]{Lemma}

\newtheorem{definition}{Definition}[section]

\newcommand{\bsmat}{\begin{bmatrix} }
\newcommand{\esmat}{\end{bmatrix} }

\newcommand{\HLinkShort}[2]{\hyperref[#2]{#1\ref*{#2}}}
\newcommand{\HLink}[2]{\hyperref[#2]{#1~\ref*{#2}}}
\newcommand{\HLinkPage}[2]{\hyperref[#2]{#1~\ref*{#2}%
		$_\text{p\pageref{#2}}$}}
\newcommand{\HLinkPageOnly}[1]{\hyperref[#1]{Page~\refpage*{#1}%
		$_\text{p\pageref{#1}}$}}

\newcommand{\HLinkSuffix}[3]{\hyperref[#2]{#1\ref*{#2}{#3}}}
\newcommand{\HLinkPageSuffix}[3]{\hyperref[#2]{#1\ref*{#2}%
		#3$_\text{p\pageref{#2}}$}}

\providecommand{\eqlab}[1]{}%
\renewcommand{\eqlab}[1]{\label{equation:#1}}

\begin{document}

\title{\bf Distances Release with Differential Privacy in Tree and Grid Graph}

\author{\vspace{0.5in}\\\textbf{Chenglin Fan} and \textbf{Ping Li} \\\\
Cognitive Computing Lab\\
Baidu Research\\
10900 NE 8th St. Bellevue, WA 98004, USA\\
  \texttt{\{chenglinfan2020,\ pingli98\}@gmail.com}
}
\date{\vspace{0.4in}}
\maketitle

\begin{abstract}

\vspace{0.4in}

	\noindent\footnote{The content of this paper was initially submitted in December 2020.}Data about individuals may contain private and sensitive information. The  differential privacy (DP)  was proposed to address the problem of protecting
		the privacy of each individual while keeping useful information about a population.
		\citet{DBLP:conf/pods/Sealfon16} introduced a private graph
		model  in which the
		graph topology   is assumed to be public while the weight information  is assumed to be private.  That model can express hidden congestion patterns in a known transportation system. In this paper, we revisit the problem of  privately releasing approximate distances between all pairs of vertices in~\citet{DBLP:conf/pods/Sealfon16}. Our goal  is to minimize the  additive error,  namely the difference between the released distance and actual distance under private setting.
		We propose improved solutions to that problem
		for several cases.\\
		
	\noindent For the problem of privately releasing all-pairs distances, we show that for  tree with depth $h$,  we can
release all-pairs distances with   additive error $O(\log^{1.5} h \cdot \log^{1.5} V)$   for fixed privacy parameter where $V$ the number of vertices in the tree, which improves the previous  error bound $O(\log^{2.5} V)$, since the size of $h$  can be as small as $O(\log V)$. Our result implies that  a $\log V$ factor is saved, and the additive  error in tree can be smaller than the error on array/path.
Additionally, for the grid graph with arbitrary edge weights, we also propose a method  to
release  all-pairs distances with additive error
$\tilde O(V^{3/4}) $ for fixed privacy parameters.   On the application side, many  cities like Manhattan are composed of  horizontal streets and vertical avenues, which can be modeled as a grid graph.
\end{abstract}

\newpage

\section{Introduction}

It has been a popular topic of research that machine learning practitioners hope to protect user privacy while effectively building machine learning models from the data. The motivation for differential privacy (DP)~\citep{DBLP:conf/pods/BlumDMN05,DBLP:conf/tcc/ChawlaDMSW05,DBLP:conf/icalp/Dwork06} is to keep useful information for model learning while protecting the privacy for individuals. The formulation of DP provides a rigorous guarantee that an adversary could learn very little about an individual.

For the graph problem under private setting, there exists a line of
works in the last decade or so~\citep{DBLP:conf/icdm/HayLMJ09,DBLP:conf/pods/RastogiHMS09,DBLP:conf/soda/GuptaLMRT10,DBLP:journals/pvldb/KarwaRSY11,DBLP:conf/tcc/GuptaRU12,DBLP:conf/innovations/BlockiBDS13,DBLP:conf/tcc/KasiviswanathanNRS13,DBLP:conf/focs/BunNSV15,DBLP:conf/pods/Sealfon16,  DBLP:conf/nips/UllmanS19,DBLP:conf/focs/BorgsCSZ18,DBLP:conf/nips/AroraU19} including node privacy, edge privacy, and weight privacy.
In this paper, we study differential privacy for the ``weight private graph model''~\citep{DBLP:conf/pods/Sealfon16}, particularly in tree and grid graph.
As the name suggests, in the weight private graph model, the topology of the graph is public but the weights are private. The weight private model can be well-suited, for example, for modeling the traffic navigation system~\citep{DBLP:conf/pods/Sealfon16}.

For two neighboring input graphs with weight functions differing by one unit, it is obvious that the  single pair shortest distance can only differ by at most one for two neighboring inputs, since the short path is a simple path with edges appearing at most once. To achieve privacy for single pairs, a popular strategy is by adding Laplace  noise according to the $O(1/\epsilon)$ based Laplace mechanism~\citep{DBLP:conf/icalp/Dwork06}. Since it is a trivial task to achieve privacy for a single pair,
\citet{DBLP:conf/pods/Sealfon16} focused on the more difficult task for releasing all-pairs distances privately.
The author showed that one can release all-pairs distances with additive error $O(\log^{2.5}
V)$ on trees, where $V$ is the number of vertices in the tree. For general graphs, the author proposed a simple approach that achieves $\tilde O(V)$ error, which was then improved when the weights are all bounded.

\vspace{0.1in}

In this paper, we revisit the problem of releasing all pairwise distances in the private graph model. Here we summarize our new results,  compared with the previous results obtained in~\citet{DBLP:conf/pods/Sealfon16}.
The additive error  is the largest absolute difference between the released distance and
the actual distance among all node pairs, which applies to both this paper and reference~\citep{DBLP:conf/pods/Sealfon16}.

\begin{itemize}
    \item
For a  tree with depth $h$, we propose  a new   algorithm to release all-pairs distances each with error  $O((\log^{1.5} h) \cdot( \log^{1.5} V ))$   for fixed privacy parameters, which is a significant improvement to previous additive error $O(\log^{2.5}
V)$~\citep{DBLP:conf/pods/Sealfon16}.
Our method is  based on heavy path decomposition~\citep{DBLP:journals/siamcomp/HarelT84}:
We divide a tree into disjoint heavy paths and light paths (the definition about heavy path decomposition is provided later in the paper). The unique path between any pair crosses at most $\log V$ heavy paths. Each heavy path is a path graph, where the releasing of  approximate all-pairs distances is equivalent to
query release of threshold functions. The results of~\citet{DBLP:conf/stoc/DworkNPR10} yield the
same error bound as  the error bound in computing distances on the path graph in~\citet{DBLP:conf/pods/Sealfon16}. We use their method as a subroutine to deal with each heavy path after heavy path decomposition~\citep{DBLP:journals/siamcomp/HarelT84}.   General graph  in metric space can be embedded  into a tree with expected  distortion $O(\log V)$ and bounded depth $O(\log V)$ by ``Padded Decomposition''~\citep{548477}, where distortion is the factor between distance/length in graph and distances in tree.  Hence the transportation network can be embedded into a depth bounded tree network.
    Our private algorithm on tree cases likely results in better private algorithms on a general graph later. On the practical side,  many internet networks are  tree networks, or star-bus networks, which can be modeled as a tree.
    Hence trees are a natural case that deserve to be studied.

\item  For a weight bounded graph $\mathbf{G}(\mathbf{V},\mathbf{E})$, the previous work~\citep{DBLP:conf/pods/Sealfon16}  picked a subset $Z\in \mathbf{V}$ of vertices to form the so called  ``$k$-covering set''. A $k$-covering set $S$ guarantees that
any vertex in $\mathbf{G}$ has   at most $k$ hops to its closest hub in $S$. One can use the $k$-covering set to approximate the original graph.  Each vertex $u$ can  map to its closest  vertex $z_u$ in the covering set.
For any pair $u,v$, the distance between $u$ and $v$ can
be approximated by the distance between $z_u$ and $z_v$ with additional error
$O(kM)$.  Thus, the solution, in general, is to use $O(|Z|^2)$ pair distances to represent/approximate the $O(V^2)$ pair distances with additional error $kM$, where $M$ is the upper bound of edge weight. As a special case, for bounded grid graphs, the authors also gave an error bound where $k=V^{1/3}$ approximately. However, it is unknown how to generalize this approach to general graphs (with arbitrary weights). To shed light on this problem, in this paper, we first consider the grid graph with general positive weights.  We divide the grid graph into blocks and then we separate the distances into several types: 1) the distances between those vertices in each block; 2) the distances between pairs of vertices on the boundary of blocks; and 3) the distances not included in  neither type 1 nor  type 2, but composed of type 1 and type 2 distances with concatenation.  Our method could
release  all-pairs distances on general grid graphs with additive error
$\tilde O (V^{3/4})$ for fixed privacy parameters; more details are given later in the paper.  We believe that our idea can be extended to more general graphs.

\end{itemize}

	\section{Background: Private Graph Model}

	For readability, we adopt the same notions as  used in~\citet{DBLP:conf/pods/Sealfon16}.
	For a general graph $\mathbf{G}=(\mathbf{V},\mathbf{E})$, throughout the paper we use $w(e)$ to denote the original weight of any edge $e\in \mathbf{E}$, and for any subset $\mathbf{E'}\subseteq \mathbf{E}$ we denote $w(\mathbf{E'}) = \sum_{e\in \mathbf{E'}} w(e)$.
	Let $V=|\mathbf{V}|$, $E=|\mathbf{E}|$, and for simplicity we assume $\mathbf{G}$ is connected and hence we always have $V-1\leq E$.

	Let $P_{xy}$ denote the set of simple paths between a pair of vertices $x, y \in \mathbf{V}$. For any path $P \in P_{xy}$, the
	weight $w(P)$ is the summation $\sum
	_{e\in P} w(e)$ of the edge weights in $P$. The distance $d(x, y)$ from $x$
	to $y$ denote the weighted distance $\min_{P \in P_{xy} } w(P)$.
	We first introduce the definition of  differential privacy (DP) in the private edge weight model~\citep{DBLP:conf/pods/Sealfon16}. The following notion of neighboring graphs will be used.
	
	\vspace{0.1in}
	
	\begin{definition}[Neighboring graphs] \label{def:neighbor}
		For any edge set $\mathbf{E}$, two weight functions $w, w'
		: \mathbf{E}\rightarrow R^{+}$      are neighboring,
		denoted $w  \sim w'$ if
		
		$$ ||w-w'||_1=\sum_{e\in \mathbf{E}} |w(e)-w'(e)| \leq 1. $$
	\end{definition}
	
	\vspace{0.1in}	
	
	\begin{definition}[Differential Privacy in graph model~\citep{DBLP:conf/pods/Sealfon16}]  \label{def:DP}
		For any graph  $\mathbf{G}=(\mathbf{V},\mathbf{E})$ let $\mathbf{A}$ be an algorithm that takes as input a weight function $w
		: E \rightarrow R^{+}$.   If for all pairs of neighboring graphs with weight $w,w'$ and for all set of outcomes $O\subseteq Range(\mathbf{A})$ such that
		$$ Pr[\mathbf{A}(w)\in O] \leq e^\epsilon Pr[\mathbf{A}(w')\in O]+\delta,$$
		algorithm $\mathbf{A}$ is said to be $(\epsilon,\delta)$-differentially private, and $\epsilon$-differentially private on $\mathbf{G}$ if $\delta=0$.

	\end{definition}
The privacy guarantees may be achieved through the introduction of noise to the output. In order to achieve $(\epsilon, 0)$-differential privacy, for example, the noise added typically comes from the Laplace distribution (the so-called Laplace mechanism will be introduced formally later). Intuitively, differential privacy requires that after removing any observation, the output of $w'$ should not be too different from that of the original data set $w$. Smaller $\epsilon$ and $\delta$ indicate stronger privacy, which, however, usually sacrifices utility. Thus, one of the central topics in the differential privacy literature is to balance the utility-privacy trade-off.

Statistically, one merit of differential privacy is that, different DP algorithms can be integrated together with provable privacy guarantee.

\begin{lemma}[Composition of DP~\citep{DBLP:conf/focs/DworkRV10}] \label{lem:ACT}
For any $\epsilon, \delta, \delta' \geq  0$, the adaptive composition of $k$ times
$(\epsilon, \delta)$-differentially private mechanisms is $(\epsilon', k\delta + \delta'
)$-differentially private for
$$ \epsilon’ = \sqrt{2k\log(1/\delta’)} \cdot \epsilon + k \cdot \epsilon(e^\epsilon-1),$$
which is $O(\sqrt{k\log (1/\delta')}\cdot \epsilon)$ when $k \leq 1/\epsilon^2$. In particular, if $\epsilon' \in (0,1),\delta'>0$, the composition of $k$ times $(\epsilon,0)$-differentially private mechanism is $(\epsilon',\delta')$-differentially private for
$$ \epsilon=\epsilon'/(\sqrt{8k \log (1/\delta')}). $$
	
\end{lemma}
	
\vspace{0.1in}
\newpage

\noindent\textbf{Private graph distance release.} In this paper, we consider the the approximate distances release problem on graphs, where the goal is to publish all the pair-wise distances (i.e., distance matrix) between all node pairs. The error is evaluated by the absolute difference between
the released/estimated distance $d_e(x,y)$ between a pair of vertices $x, y$ and the actual distance $d(x, y)$. For each pair of vertices $(x,y)$, we call $|d_e(x,y)-d(x,y)|$ the additive error of that pair. The objective is to  minimize the largest additive error among all pairs, namely, minimizing $\{ \max \{ |d_e(x,y)-d(x,y)|,(x,y)\in \mathbf{E} \} \}$ under the  constraints  that all $\{d_e(x,y), (x,y)\in \mathbf{E}\}$ are differentially private achieving Definition~\ref{def:DP}.

\vspace{0.1in}
\noindent\textbf{Technical tools.} We now introduce a few technical tools which will be used throughout the remainder of this paper.
A number of differential privacy techniques incorporate noise sampled according to the Laplace
distribution. The so-called Laplace mechanism will be frequently used in our algorithm design and analysis.
	
\vspace{0.1in}	
\begin{lemma}[Laplace mechanism~\citep{DBLP:conf/icalp/Dwork06}] \label{lem:basic}
For a function $f:\mathcal G\rightarrow \mathbb{R}$ with $\mathcal G$ the input space of graphs, define the $l_1$ sensitivity as
$$\triangle_f=\max_{G \sim G'} || f(G)-f(G')||_1,$$
where $G,G'$ are two neighboring graphs as in Definition~\ref{def:neighbor}. Let $\xi$ be a random noise drawn from $Lap(0,\triangle_f/\epsilon)$. The Laplace mechanism outputs
$$ M_{f,\epsilon}(G)  = f(G)+\xi,$$
and the approach is $\epsilon$-differentially private.
\end{lemma}

We state a concentration bound for the summation of Laplace
random variables. Even though  these results are already well known (and used heavily in the prior work~\citep{DBLP:conf/pods/Sealfon16}), we just include them for completeness.

\vspace{0.1in}

\begin{lemma}[Concentration of Laplace RV~\citep{DBLP:conf/icalp/ChanSS10}\label{lem:bound}] Let $Z$ be the sum of $n$ i.i.d. random variables $Z_1,Z_2,...,Z_n$ following $Lap(b)$, and denote $Z=\sum_1^n Z_i$. For $0<t<2\sqrt 2 bn$, we have
$$Pr[|Z|>t] \leq 2\exp(-\frac{t^2}{8nb^2}).$$
For $\gamma\in(2/e^n,1)$, with probability at least $1-\gamma$,
$$ |Z|<2\sqrt 2 b\sqrt{n} \log (2/\gamma)=O(b\sqrt{n} \log (1/\gamma)).$$
\end{lemma}

	\vspace{0.05in}
	
	\section{
		Distances Release in Trees}
	
	The notion of depth bounded tree is widely used in computer science. For example, general graphs can be embedded into a tree with bounded depth $O(\log V)$~\citep{548477}, and improvements on private trees would likely lead to better private algorithms on general graphs.
	Since the previous paper~\citep{DBLP:conf/pods/Sealfon16}  could deal with the tree case with additive error $O(\log^{2.5} V$ for fixed  privacy parameters, it is natural to ask whether that bound can be improved. The answer is ``yes'',  and this paper provides a method which achieves better performance
	for the depth bounded tree.

	\newpage

	\begin{figure}[t]

		\centering
		\includegraphics[width=2.5in]{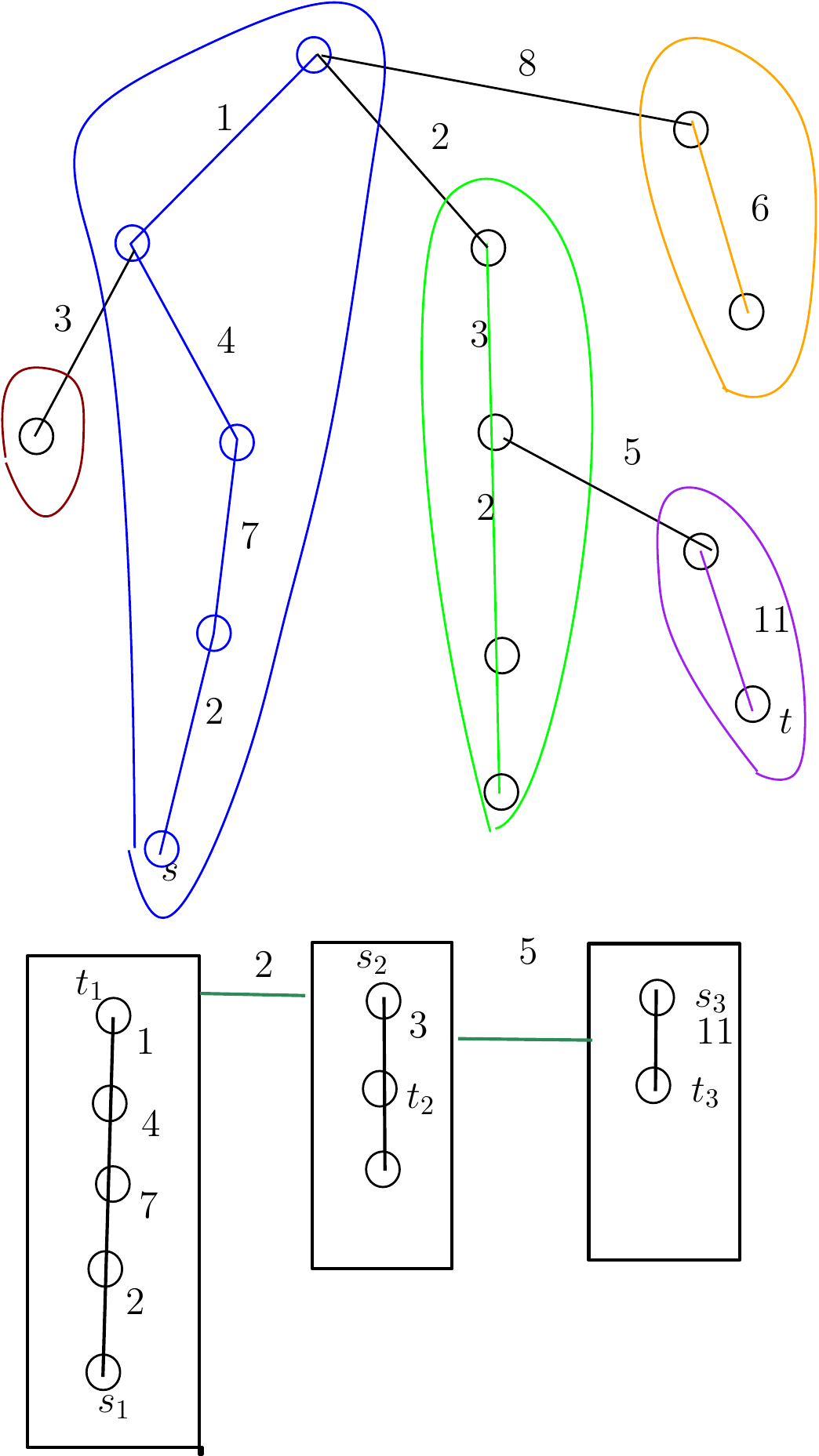}

		\caption{ The input tree $\mathbf{T}$ is partitioned into paths by classic heavy path decomposition. Each heavy path is marked with a distinct color. The shortest path between $s$ and $t$ can be decomposed into sub paths like $(s_1,t_1)$,$(s_2,t_2)$, $(s_3,t_3)$ inside disjoint heavy paths and light edge between them.}
		\label{heavy}
	\end{figure}

	We can decompose the tree into paths using the heavy path decomposition~\citep{DBLP:journals/siamcomp/HarelT84}, which is also called heavy-light decomposition and is a technique for decomposing a rooted tree into a set of paths as long as possible; see Algorithm~\ref{heavy_path_decomposition}. In a heavy path decomposition, each non-leaf node selects one branch, the edge to the child that has the largest depth (breaking ties arbitrarily). The selected edges form the paths of the decomposition, an example is given in Figure \ref{heavy}.
	Since the tree is decomposed into a set of paths, some edges may not be included in any one of the
	heavy paths produced during the decomposition process. We call those edges the ``light edge'', in comparison to those heavy edges included in the heavy~paths.
	
\vspace{0.1in}

	\begin{lemma}[Tree Decomposition~\citep{DBLP:journals/siamcomp/HarelT84}] \label{lem:light}
		For any root-to-leaf path of a tree with $V$ nodes, there can be at most $\log V$ light edges. Equivalently, the path tree has height at most $\log V$.
	\end{lemma}
	
	We now present the general idea of our algorithm. We first decompose the tree into heavy paths. Each heavy path is a
	path graph, which we can use the classic private algorithm in~\citet{DBLP:conf/pods/Sealfon16,DBLP:conf/stoc/DworkNPR10} to deal with. Also, these heavy paths are all disjoint to each other. Thus each of them can be handled separately. For those light edges, each of them can be added with a random Laplace noise according to the random variable $Lap(1/\epsilon)$.

	The query process: for a pair $(s,t)$ inside one heavy path, which is similar to the query process of path graphs. For another pair $(s,t)$ crossing several heavy paths, the path distance can be released by summing several subqueries of heavy paths and light edges.

	\begin{algorithm2e}
		\SetKwInOut{Input}{Input}
		\SetKwInOut{Output}{Output}
		\Input{ Tree  $\mathbf{T}$.}
		\Output{A set of edge disjoint paths by heavy path decomposition.}
		\DontPrintSemicolon

		Let $\mathbf{R}$  denote a set of roots of subtrees in $\mathbf{T}$.\;

		Let $P$ denote a path (a set of consecutive nodes in $\mathbf{T}$ ).

	\If {$\mathbf{T}$ is a leaf }
	{
	Add $\mathbf{T}$ to $P$.\;
	}
		Let node pointer $p$ point to the root of $\mathbf{T}$.\;

		\While {$p$ is not null}
		{
			Add $p$ to $P$.
			\;
		Compute the depth of each children of $\mathbf{T}$.\;
		Let $p$ point to the children of $\mathbf{T}$ with maximum depth.\;
		Add another children nodes of $\mathbf{T}$  than $p$ to $\mathbf{R}$.
		}

		\For {node $v$ in  $\mathbf{R}$}
		{
			Call  recursive Tree-Decomposition$(v)$.\; Add the obtained paths into $S$.\;

		}

		Let $S=\{ P_1,P_2,...,P_k\}$ be a set of heavy paths produced in the decomposition above.\;
		Return    $S$.\;
		\caption{  Tree-Decomposition$(\mathbf{T})$: To partition a tree $\mathbf{T}$ into a set of edge disjoint paths.}
		\label{heavy_path_decomposition}
	\end{algorithm2e}

	\begin{algorithm2e}
		\SetKwInOut{Input}{Input}
		\SetKwInOut{Output}{Output}
		\Input{Length bounded Tree  $\mathbf{T}$, private  parameter $\epsilon$.}
		\Output{Private all pairwise distances of vertices in  $\mathbf{T}$.}
		\DontPrintSemicolon
			Call Tree-decomposition algorithm in  Algorithm~\ref{heavy_path_decomposition} on $\mathbf{T}$. \;
		Let $S=\{ P_1,P_2,...,P_k\}$ be a set of heavy paths produced in the decomposition above.\;
		
		\For {each heavy path $P_i$ in S}
		{
			Call  the Algorithm 1 in~\citet{DBLP:conf/pods/Sealfon16} to release private all pairwise distances  of path $P_i$ with privacy parameter $\epsilon$.\;

		}
		
		\For {each light edge $(u,v)$  }
		{
			$d_e(u,v):= d(u,v)+$Lap$(1/\epsilon)$.  \;
		}
		
		\For { each pair $(s,t)$ }
		{
			Find the path $P(s,t)$  from  vertex $s$ to $t$.\;

			The shortest path between $s$ and $t$ can be decomposed into sub paths like $(s_1,t_1)$, $(s_2,t_2)$, $(s_3,t_3),...,$ into  disjoint heavy paths.\;
			\For {each heavy path intersecting with $P(s,t)$ }
			{
				Add the released distance $d_e(s_i,t_i)$ to $d_e(s,t)$.\;
			}
			\For {each light edge intersecting with $P(s,t)$ }
			{
				Add the release distance of that
				light edge to $d_e(s,t)$.\;
			}

		}
		
		Return    all pairwise  distances released   above.\;
		\caption{Private algorithm to release all pairwise distances  for  tree.}
		\label{Path_Tree}
	\end{algorithm2e}

\newpage

	\begin{theorem}[All Pairwise Shortest Path Distances on Rooted Trees]
		\label{thm:all_pair}
			Let $\mathbf{T} = (\mathbf{V}, \mathbf{E})$ be a tree with $V$ vertices, $\epsilon> 0$, Algorithm \ref{Path_Tree} is $\epsilon$-differentially private on $\mathbf{T}$ and 		releases
all-pairs distances such that with probability $1-\gamma$ , all released distances have additive error bounded by
	$$O((\log^{1.5} V\cdot \log^{1.5}h)\log(1/\gamma))/\epsilon,$$
	which is
		$$O(\log^{1.5} V \cdot \log\log^{1.5} V \cdot \log (1/\gamma))/\epsilon$$
		when $h=O( \log V).$
	\end{theorem}
	
	\begin{proof}
		Since each of heavy paths and light edges are disjoint to each other, we can assign $\epsilon$
		privacy to each of them, which gives $\epsilon$-DP in total.
		Let $S$ be a set of disjoint heavy paths of $T$.
		Let $D_i$ denote the length of $i$-th heavy path of $S$ ($|S|\leq \log V$), and $d_e(s_i,t_i)$ be the sub query inside the $i$-th path needed for the pair $(u,v)$. The length of each heavy path is bounded by the depth $h$ of the tree.
	The additive error of each pair $d_e(s_i,t_i)$ inside path $(s_i,t_i)$ is the sum of $\log L(s_i,t_i)$ Laplace random variable according to $Lap(\log h/\epsilon)$ based on~\citet{DBLP:conf/pods/Sealfon16} and the
	link length of each heavy path is bounded by $h$, and $L(s_i,t_i)$ denotes the link length between $s_i$ and $t_i$.
		
		For those heavy paths between $s$ and $t$,  the additive error of each of them is decided by sum of   $\log L(s_i,t_i)$ Laplace random variables each  following $Lap(\log h/\epsilon)$.
		Hence, the sum of additive error of those heavy paths
		is determined by  the sum of  $\sum \log L(s_i,t_i)$ Laplace random variables each according to $Lap(\log h/\epsilon)$.  We know that $\sum L(s_i,t_i) \leq 2h$ because  the link length of any path is less or equal than $2h$.
		With probability $1-\gamma$ this error is bounded by
		$$O(\sqrt{\sum \log L(s_i,t_i) } (\log h) \log(1/\gamma) )/\epsilon$$
		based on Lemma~\ref{lem:bound}.
	As the path between $(s,t)$ passes at most $\log V$ light edges based on Lemma~\ref{lem:light}, with probability $1-\gamma$, the error of this part is bounded by
		$$O(\sqrt{\log V}\log(1/\gamma))/\epsilon$$
	based on	Lemma~\ref{lem:bound}.
		Therefore, the total sum of additive error between $(s,t)$ is bounded  by
		$$O(\sqrt{\sum (\log L(s_i,t_i))}  \log h+\sqrt{\log V})\frac{\log(1/\gamma)}{\epsilon}.$$
		On the other hand, ${\sum \log L(s_i,t_i) }$ can be bounded by  $O((\log V)\log h)$ as $L(s_i,t_i)\leq h$ and the path from $s$ to $t$ passes at most $O(\log V)$ heavy paths.
		Combining parts together, the additive error is bounded  by
		$$O(\log^{0.5} V\log^{1.5} h  +\log^{0.5} V)\frac{\log(1/\gamma)}{\epsilon}.$$
	By a union bound, for  any $\gamma \in (0, 1)$, with
probability at least $1-\gamma$, each error among the $O(V^2)$  all-pairs distances released is at most
$O(\log^{1.5} V\cdot \log^{1.5}h )\frac{\log(1/\gamma)}{\epsilon}.$
	\end{proof}

\section{Distances Release in Grid Graph}
	
	For a general graph $\mathbf{G}(\mathbf{V},\mathbf{E})$, one can pick a subset $U\in \mathbf{V}$ of vertices, the so called  ``$k$-covering set'', to approximate the graph when the weight is bounded. That approach could not be extended to more general settings, whereas in this paper we consider the distance release  in grid graph with arbitrary weights.

	


	
 The general idea is as follows. Let $\mathbf{G}$ be the $\sqrt{V} \times \sqrt{V}$
	grid. A path from $s$ to $ t$ can be divided into three parts by two immediate vertices $(u,v)$:  $(s,...,u),(u,..., v), (v,...,t)$.
 In order to obtain the immediate vertex set, 	
	we divide the grid into $\sqrt{V}$ blocks  with size $V^{1/4} \times V^{1/4}$.
	Let $B_1,B_2,...,B_{\sqrt{V}}$ denote the $\sqrt{V}$ blocks respectively.
	Let set $A_i$ denote the set of  vertices located on the boundary of  $B_i$, and the size of $A_i$ is
	$O(V^{1/4})$.
	Let set $A=\bigcup_{i}  A_i$, then $A$ is the immediate vertex set for grid graph.
 The additive error of $d_e(s,u)$ is decided by the number of edges from  $s$ to $u$, and similar analysis for $d_e(v,t)$. The additive error of $d_e(u,v)$ is  the noise added to $d(u,v)$, which depends on the size of the immediate vertex set $A$, specified next.

	For each pair vertices $(u,v)$  in  $A$,   we add Laplace random noise
	$Lap(V^{3/4}( \sqrt{8\log 1/\delta})/\epsilon)$ to the released distance, namely $d_e(u,v):=d(u,v)+Lap(V^{3/4}( \sqrt{8\log 1/\delta})/\epsilon)$, recall that $d(u,v)$ is the exact distance between $u$ and $v$.
	For each pair vertices $(x,y)$ in  each $B_i$, add Laplace random noise
	$Lap(V^{3/4}( \sqrt{8\log 1/\delta})/\epsilon)$ to the released distance, namely $d_e(x,y):=d(u,v)+Lap(V^{3/4}( \sqrt{8\log 1/\delta})/\epsilon)$.
	For a pair $(s,t)$ such that $s$ is located in $B_i$ while $t$ is  located in some $B_j$. We release
	$d_e(s,t)$ as follows.
	$$d_e(s,t):=\min_{u\in A_i,v \in A_j} \{ d_e(s,u)  +d_e(u,v)+d_e(v,t) \}.$$


	An illustration is given in Figure  \ref{distance}. In Theorem~\ref{thm:graph}, we show that for any $\epsilon, \gamma \in (0, 1)$ and $\delta > 0$, one can release
	with probability $1 -\gamma$ all-pairs distances each with additive  error $O(V^{3/4}\log (V/\gamma)  \sqrt{\log 1/\delta})/\epsilon$.
	
	\begin{figure}[h]
\vspace{0.1in}

		\centering
		\includegraphics[width=2.2in]{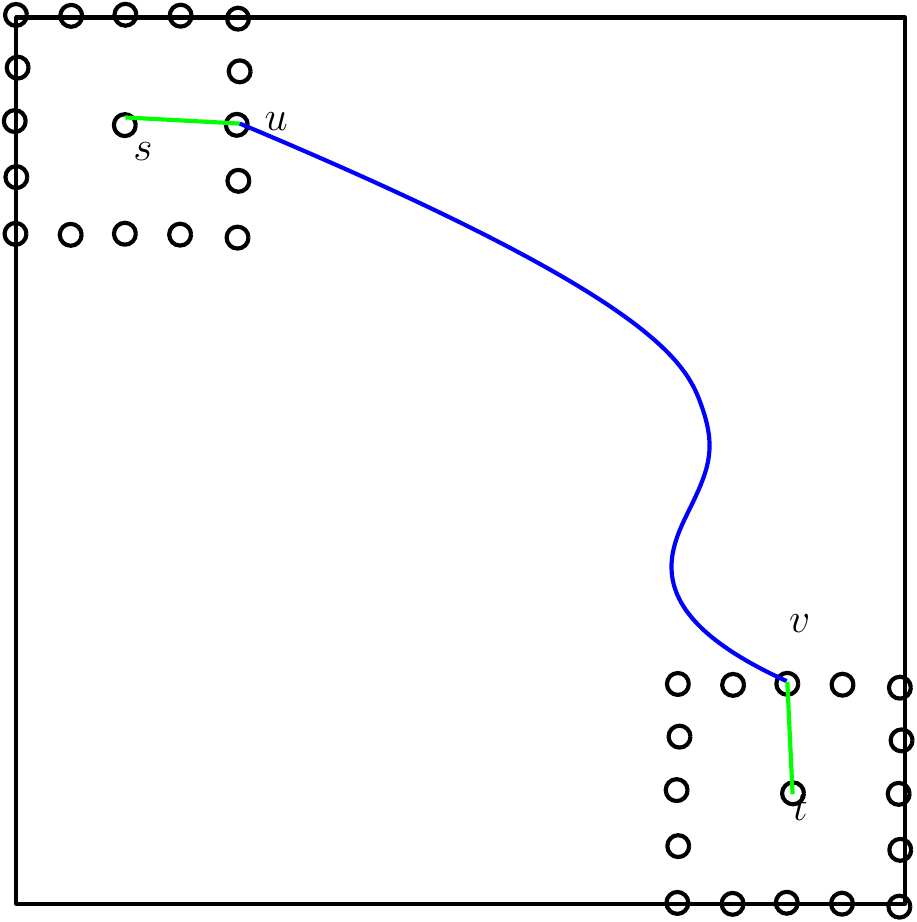}
		\caption{Illustration of released distance  between $(s,t)$ in two different blocks.}
		\label{distance}
	\end{figure}

	\begin{algorithm2e}[h]
		\SetKwInOut{Input}{Input}
		\SetKwInOut{Output}{Output}
		\Input{An instance $\mathbf{G}(\mathbf{V},\mathbf{E})$ , distances  $\{w(e),e\in \mathbf{E}\}$, parameter $\epsilon,\delta$.}
		\Output{All pairwise distances of $\mathbf{G}$.}
		\DontPrintSemicolon

		Divide the grid into $V$ blocks  with size $V^{1/4} \times V^{1/4}$.
		Let $B_1,B_2,...,B_{V^{1/2}}$ denote the $V^{1/2}$ blocks respectively.\;
		Let set $A_i$ denote a set of vertices located on the boundary of  $B_i$.\;
		Let set $A=\bigcup_{i}  A_i$.\;
		
		\For { each  $u, v \in A$}
		{
			Compute the exact distance $d(u,v)$. \;
			
			Let the released distance between $u$ and $v$ be $d_e(u,v):= d(u,v)+ Lap(V^{3/4} \sqrt{8\log 1/\delta}/\epsilon)$.\;
			
		}
		
		\For {each block $B_i$}
		{
			
			\For { each $(u,v) \in B_i$ }
			
			{
				Compute the exact distance $d(u,v)$. \;
				
				Let the released distance between $u$ and $v$ be $d_e(u,v):=  d(u,v)+ Lap(V^{3/4} \sqrt{8\log 1/\delta}/\epsilon)$.\;

			}
		}
		
		\For { each $(s,t)$ in two distinct blocks $B_i,B_j$}
		{
			
			$d_e(s,t):= \min_{u\in A_i,v \in A_j} \{ d_e(s,u)  +d_e(u,v)+d_e(v,t) \} $\;
		}
		Return  $d_e(s,t)$.\;
		\caption{Private Algorithm for Grid Graph.}
		\label{graph}
	\end{algorithm2e}

	
	\begin{theorem}\label{thm:graph}
		Given any input $\mathbf{G} ($  a grid graph  $\sqrt{V} \times \sqrt{V})$
		, and parameter $\epsilon,\delta>0$, Algorithm~\ref{graph}  is $(\epsilon,\delta)$-differentially private.  With probability $1-\gamma$,
		   Algorithm~\ref{graph} releases all pairwise distances,
		   for each distance,
		   with probability $1-\gamma$ the additive error is  $O(V^{3/4}( \sqrt{\log 1/\delta}) \log(V/\gamma))/\epsilon.$
	\end{theorem}
	
	\newpage
	
	\begin{proof}
		For the privacy part, only  $O(V^{3/2})$ pairs of distances are computed directly, other pairs are composed of by them. Hence adding noise according to  $ Lap(V^{3/4}(\sqrt{8\log 1/\delta})/\epsilon)$ suffice to guarantee privacy based on the Lemma~\ref{lem:ACT} and Laplace Mechanism.
		
		For the utility part, the pair  $(s,t)$ in two distinct blocks $B_i,B_j$,  the shortest path between $B_i,B_j$ has to pass two points $u,v$ such that  $u\in A_i$	and $v\in A_j$.
		Hence 	$d_e(s,t)=\min_{u\in A_i,v \in A_j} \{ d_e(s,u)  +d_e(u,v)+d_e(v,t) \} $ can get the shortest path distance  between $s$ and $t$.
		With probability $1-\gamma$,  each of  the total $V^{3/2}$  Laplace random variables according to $O(V^{3/4}(\sqrt{\log 1/\delta})/\epsilon)$ can be bounded by $O(V^{3/4}(\sqrt{\log 1/\delta})\log (V/\gamma)/\epsilon)$. The sum of three variables increase the error by at most three times, which can be ignored.
		Hence, for any $\gamma  \in  (0, 1)$, with probability
		at least $1-\gamma$, the additive error of $d_e(s,t)$ is within
		$O(V^{3/4}\log (V/\gamma)  \sqrt{\log 1/\delta})/\epsilon$.
	\end{proof}

	\section{Conclusion}
	
	In this paper, we study the problem of releasing all pairwise distances in the private graph model as studied in~\citet{DBLP:conf/pods/Sealfon16}.
	For the problem of privately releasing all-pairs distances in trees,
the author proposed a solution to achieve  additive error $O(\log^{2.5}
V)$ for fixed privacy parameters, where $V$ is the number of vertices in the tree.  In this paper, we propose  a new   algorithm  which can release all-pairs distances with error  $O(\log^{1.5} h \cdot \log^{1.5} V )$   for fixed privacy parameters.
	Our method is  based on heavy path decomposition~\citep{DBLP:journals/siamcomp/HarelT84}, and $h$ is small in many applications.   Additionally, in this paper, we also consider the grid graph with general positive weights. Our approach is based on dividing the graph into blocks and selecting an intermediate set of vertices in distance computation. Our proposed method
	releases all-pairs distances with additive error
	$\tilde O(V^{3/4})$ for fixed privacy parameter on general grid graphs.

\bibliographystyle{plainnat}	
\bibliography{dp}

\end{document}